\documentclass[12pt]{amsart}
\usepackage{amsmath}
\usepackage{amsxtra}
\usepackage{amscd}
\usepackage{amsthm}
\usepackage{amsfonts}
\usepackage{amssymb}
\usepackage{eucal}
\usepackage{epsfig}
\usepackage{graphics}
\newcommand{\bb}{\mathbf{b}}
\newcommand{\cb}{\mathbf{c}}
\newcommand{\tb}{\mathbf{t}}
\newcommand{\xb}{\mathbf{x}}
\newcommand{\hb}{\mathbf{h}}

\newcommand{\nn}{\nonumber}
\newcommand{\bea}{\begin{eqnarray}}
\newcommand{\ena}{\end{eqnarray}}
\def\bel{\begin{eqnarray}}
\def\enl{\end{eqnarray}}
\newcommand{\be}{\begin{eqnarray*}}
\newcommand{\en}{\end{eqnarray*}}
\newcommand{\ba}{\begin{array}}
\newcommand{\ea}{\end{array}}

\newcommand{\C}{{\mathbb C}}
\newcommand{\Z}{{\mathbb Z}}

\newcommand{\slth}{\widehat{\mathfrak{sl}}_2}

\newcommand{\End}{\mathop{\rm End}}
\newcommand{\la}{\lambda}
\newcommand{\al}{\alpha}
\newcommand{\z}{\zeta}
\newcommand{\e}{\epsilon}
\numberwithin{equation}{section}
\newtheorem{thm}{Theorem}[section]
\newtheorem{prop}[thm]{Proposition}
\newtheorem{lem}[thm]{Lemma}

\newtheorem{cor}[thm]{Corollary}

\newcommand{\cW}{\mathcal{W}}
\newcommand{\cF}{\mathcal{F}}
\newcommand{\Wal}{\cW^{(\al)}}
\newcommand{\bS}{\mathbb{S}}

\newcommand{\supp}{\mathop{\rm supp}}
\newcommand{\Rcb}{{\widetilde{\mathbb{R}}^\vee}}

\begin{document}

\begin{title}[Completeness of a fermionic basis]
{Completeness of a fermionic basis 
in the homogeneous XXZ model}
\end{title}
\date{\today}
\author{H.~Boos, M.~Jimbo, T.~Miwa and  F.~Smirnov}
\address{HB: Physics Department, University of Wuppertal, D-42097,
Wuppertal, Germany}\email{boos@physik.uni-wuppertal.de}
\address{MJ: Graduate School of Mathematical Sciences, The
University of Tokyo, Tokyo 153-8914, Japan}
\email{jimbomic@ms.u-tokyo.ac.jp}
\address{TM: Department of 
Mathematics, Graduate School of Science,
Kyoto University, Kyoto 606-8502, 
Japan}\email{tetsuji@math.kyoto-u.ac.jp}
\address{FS\footnote{Membre du CNRS}: Laboratoire de Physique Th{\'e}orique et
Hautes Energies, Universit{\'e} Pierre et Marie Curie,
Tour 16 1$^{\rm er}$ {\'e}tage, 4 Place Jussieu
75252 Paris Cedex 05, France}\email{smirnov@lpthe.jussieu.fr}

\begin{abstract}
With the aid of the creation operators introduced in our previous works, 
we show how to construct a basis of the space 
of quasi-local operators for the homogeneous XXZ chain.
\end{abstract}

\maketitle

\bigskip


\section{Introduction}\label{Intro}

The present note is a supplement to our previous papers \cite{HGS}--\cite{HGSIII} 
on the fermionic structure in the XXZ model. 
Our aim here is to give a proof of the completeness statement
announced in \cite{HGSIII} 
concerning the fermionic basis of quasi-local operators. 

The space of states of the XXZ model on the infinite lattice 
is formally the tensor product $\otimes_{j=-\infty}^\infty V_j$, 
where $V_j$ is a copy of the spin $1/2$ evaluation module
of $U'_q(\slth)$. Of interest is the expected values of local operators, 
i.e., those which act as identity on $V_j$ for all but a finite number of $j$. 
A quasi-local operator $X$ is one which acts on almost all $V_j$ as 
\be
q^{2(\al-s) S(0)},\quad 
S(k)=\frac{1}{2}\sum_{j=-\infty}^k \sigma^3_j\,.  
\en
Here $q^\al$ is a parameter and $s\in \Z$ is the spin 
of $X$. In other words, $X$ is quasi-local with spin $s$ if 
\bea
X=q^{2(\al-s) S(0)}\mathcal{O}  
\label{quasi-local}
\ena
where $\mathcal{O}$ is  local and $[S(\infty),\mathcal{O}]=s\, \mathcal{O}$. 
The space of all quasi-local operators \eqref{quasi-local} with spin $s$ 
is denoted by  $\cW^{(\al)}_s$.     

In \cite{HGS,HGSII}, we have introduced 
a set of linear operators 
\bea
\tb^*_p,\, \bb_p,\, \cb_p,\, \bb^*_p,\, \cb^*_p\,
\quad(p\ge 1),
\label{tbcbc}
\ena
which act on $\cW^{(\al)}=\oplus_{s\in\Z}\cW^{(\al)}_s$.  
Among them, $\bb_p,\cb_p$ annihilate the 
`vacuum' $q^{2\al S(0)}$, while
the others can be used to create a family of 
elements of $\cW^{(\al)}$:
\bea
&&(\tb^*_1)^p \tb^*_{i_1}\cdots\tb^*_{i_r}
\bb^*_{j_1}\cdots\bb^*_{j_s}
\cb^*_{k_1}\cdots\cb^*_{k_t}(q^{2\al S(0)})
\label{monomial}\\
&&
(i_1\ge\cdots\ge i_r\ge 2,\ j_1>\cdots>j_s\ge 1, \
k_1>\cdots>k_t\ge 1,\ p\in\Z,\ r,s,t\ge 0)\,.
\nn
\ena
The main result of \cite{HGSII} states that, 
at zero-temperature, 
the normalized expected values of the generating functions of 
\eqref{monomial} can be described explicitly in terms 
of determinants. This result has been generalized 
in \cite{HGSIII} to the case of a non-zero temperature. 
In principle, one has therefore 
a complete knowledge about the expected value
of an arbitrary quasi-local operator,  
provided that \eqref{monomial} gives a basis of $\cW^{(\al)}$. 

Using the (anti-)commutation relations among 
\eqref{tbcbc}, 
it is fairly straightforward to show that the elements 
 \eqref{monomial} are linearly independent 
 (see Corollary \ref{independence2} below). 
Hence the issue is that of the completeness, i.e., 
whether \eqref{monomial} span the whole space $\cW^{(\al)}$. 
We show here that this is indeed the case. 
Our proof rests on the 
commutation relations whose proof is partly unfinished, 
so the completeness is settled modulo this point.

The plan of the text is as follows. 
In Section \ref{sec2}, we give a summary 
of known facts about the operators \eqref{tbcbc}. 
We also introduce an alternative set of creation operators
 $\overline{\bb}^*_p,\overline{\cb}^*_p$.  
While $\bb^*_p,\cb^*_p$ enlarge the support of the operand to the right, 
 $\overline{\bb}^*_p,\overline{\cb}^*_p$ do the same but to the left. 
Using these operators, 
we construct in Section \ref{sec3} a set of elements 
supported on a given interval, 
and show that they span the space of quasi-local operators. 

Throughout the text, 
$q$ is a non-zero complex number.  
Except in Lemma \ref{independence}
we assume that $q$ is not a root of unity. 
We work over the base field $\C(q^\al)$, 
$q^\al$ being an indeterminate.  

\section{Preliminaries}\label{sec2}

In this section we recall the main features of
the creation and annihilation operators 
\eqref{tbcbc}.  
Their construction is rather involved  
and occupies a large part of \cite{HGSII}. 
Leaving the details to that paper, we give below 
a summary of their properties.

\subsection{Commutation relations}

In the algebra generated by \eqref{tbcbc}, 
$\tb^*_p$ are central: 
\bea
[\tb^*_p,\xb_{p'}]=0\quad 
(p,p'\ge1,\  \xb=\tb^*, \bb,\cb, \bb^*,\cb^*)\,.
\label{Com1}
\ena

The rest of the operators obey the canonical anti-commutation 
relations 
\bea
&&[\bb_p,\bb_{p'}]_+=[\bb_p,\cb_{p'}]_+=[\cb_p,\cb_{p'}]_+=0,
\label{Com2}
\\
&&[\bb^*_p,\bb_{p'}]_+=
[\cb^*_p,\cb_{p'}]_+=\delta_{p,p'},\quad 
[\bb^*_p,\cb_{p'}]_+=[\cb^*_p,\bb_{p'}]_+=0\,,
\label{Com3}
\\
&&
[\bb^*_p,\bb^*_{p'}]_+
=[\bb^*_p,\cb^*_{p'}]_+=[\cb^*_p,\cb^*_{p'}]_+=0\,.
\label{Com4}
\ena
The relations \eqref{Com1}--\eqref{Com3} are proved in
\cite{HGSII}. 
At this writing, proof of the 
last set of relations \eqref{Com4} remains 
unavailable. 
In what follows we shall assume \eqref{Com4}.  

\subsection{Support property}

For an element $X\in \cW^{(\al)}_s$, 
its support $\supp X$ is the minimal interval $[k,m]$
outside which $X$ coincides with $q^{2(\al-s) S(k-1)}$.  
We denote by $(\cW^{(\al)})_{[k,m]}$ 
the subspace of $\cW^{(\al)}$ 
consisting of elements $X$ such that
$\supp X\subset [k,m]$. 
Sometimes we abuse the language and say that 
a local operator $\mathcal{O}$ is supported on 
$[k,m]$ if it acts as identity outside $[k,m]$.  
The operators \eqref{tbcbc} 
respect the support of the operand in the following sense. 

The operators $\bb_p,\cb_p$ preserve the support.
More specifically, we have 
for any $X\in (\cW^{(\al)})_{[k,m]}$ 
\bea
&&\supp \xb_p(X)\subset [k,m]
\quad (\xb=\bb,\cb),
\label{supp1}\\
&&\xb_p(X)=0\quad \text{if $p>m-k+1$} 
\quad (\xb=\bb,\cb).
\label{kill}
\ena
The second property justifies to call them annihilation operators. 

In contrast, the creation 
operators $\tb^*_p, \bb^*_p,\cb^*_p$ enlarge the support
according to the rule
\bea
&&\supp \xb^*_p(X)\subset [k,m+p]
\quad (\xb^*=\tb^*, \bb^*,\cb^*).
\label{supp2}
\ena
Notice that the support is enlarged {\it only to the right}.  

\subsection{Independence of monomials}

Among the operators $\tb^*_p$, 
$\boldsymbol{\tau}=\tb^*_1/2$ plays a special role: 
it is the translation to the right by one lattice unit. 
In particular \eqref{Com1} tells 
that all operators are translationally invariant. 

Separating $\tb^*_1$ from the rest, we set 
\be
\hb^*_p=(\tb^*_1)^{-1}\tb^*_{p+1},\quad
\hb^*(\z)=\sum_{p=0}^\infty (\z^2-1)^{p}\hb^*_p\,.
\en
Accodring to \cite{HGSII}, Subsection 3.4, 
$\hb^*(\z)$ is given by the following formula. 
For $X\in (\cW^{(\al)})_{[k,m]}$, we have  
\bea
&&
\hb^*(\z)(X)=\lim_{l\to\infty}
\Rcb_{l,l-1}(\z^2)\cdots \cdots\Rcb_{k,k-1}(\z^2)(X).
\label{hbz}
\ena
Here $\Rcb_{i+1,i}(\z^2)\in \End(V_i\otimes V_{i+1})$
denotes the adjoint action by the standard $R$ matrix 
for the XXZ model. 
We have $\Rcb_{i+1,i}(1)=1$. Moreover 
if $\supp(X)\subset[a,b]$ with either $i>b$ or $i<a-1$, then 
$\Rcb_{i+1,i}(\z^2)(X)=X$ holds. 
These properties of $\Rcb_{i+1,i}(\z^2)$ ensure that 
each coefficient of the series 
appearing in the right hand side of \eqref{hbz} 
stabilizes for large enough $l$. 
More precisely, we have for any $p$ 
\be
\hb^*(\z)(X)\equiv 
\Rcb_{m+p,m+p-1}(\z^2)\cdots \cdots\Rcb_{k,k-1}(\z^2)(X)
\quad \bmod (\z^2-1)^{p+1}. 
\en
Consequently we have the support property
\begin{equation}
\supp \hb^*_p(X)\subset[k-1,m+p] \,.
\label{hsupp}
\end{equation}

\begin{lem}\label{independence}
If $q$ is generic, then the set of elements
\bea
(\hb^*_1)^{m_1}(\hb^*_2)^{m_2}\cdots(q^{2\al S(0)}) 
\quad (m_1,m_2,\cdots\ge 0)
\label{tmono}
\ena
is linearly independent. 
\end{lem}
\begin{proof}
It is enough to show the assertion for a special value of 
the parameter $q$. 
We choose $q=\sqrt{-1}$, 
writing $q^\al$ as $y$ and $q^{2\al S(0)}$ as $y^{2S(0)}$. 

Following \cite{HGS}, let us introduce linear operators 
\begin{equation*}
\Psi^\pm_j,\ \Phi^\pm_j~:~\Wal_s\longrightarrow \Wal_{s\pm1} 
\end{equation*}
by the formula
\begin{align}
\Psi^\pm_j(X)&=\psi^\pm_j X-(-1)^sX\psi^\pm_j,
\\
\Phi^\pm_j(X)&=\frac{1}{1-y^{\mp 2}}
\left(\psi^\pm_j X-y^{\mp 2}(-1)^sX\psi^\pm_j\right),
\end{align}
where $X\in \Wal_s$ and  
\begin{equation*}
\psi^\pm_j=\sigma^\pm_j e^{\mp \pi i S(j-1)}\,.
\end{equation*}
They satisfy the canonical anti-commutation relations
\begin{equation*}
[\Psi^\e_j,\Psi^{\e'}_k]_+=0,\quad
[\Phi^\e_j,\Phi^{\e'}_k]_+=0,\quad
[\Psi^\e_j,\Phi^{\e'}_k]_+=\delta_{j,k}\delta_{0,\e+\e'},
\end{equation*}
as well as the annihilation property
\begin{align*}
&\Psi^\pm_j(y^{2 S(0)})=0\quad (j>0),
\\
&\Phi^\pm_j(y^{2 S(0)})=0\quad (j\le 0)\,.
\end{align*}
Acting with the fermions $\Psi^\pm_j$, $\Phi^\pm_j$ 
over the vacuum $y^{2 S(0)}$ 
we obtain a Fock space $\cF\subset \Wal$. 

When $q=\sqrt{-1}$, the $R$ matrix simplifies to 
\begin{align*}
&\Rcb_{i+1,i}(\z^2)(X)=e^{H_{i,i+1}} X e^{-H_{i,i+1}},
\\
&H_{i,i+1}=z 
(\psi^+_{i}+\psi^+_{i+1})(\psi^-_{i}-\psi^-_{i+1})\,,
\end{align*}
where
\begin{equation*}
z=\frac{1-\z^2}{1+\z^2}\,.
\end{equation*}
Expressing the left and right multiplication operators by
$\psi^\pm_j$ in terms of $\Psi^\pm_j$ and $\Phi^\pm_j$, 
we obtain an identity of linear operators on $\cF$
\begin{equation*}
\Rcb_{i+1,i}(\z^2)=
\exp\left(z\sum_{\e=\pm}
(\e\ \Phi^\e_i+\Phi^\e_{i+1})
(\Psi^{-\e}_i-\e \Psi^{-\e}_{i+1})\right)\,.
\end{equation*}
A simple calculation using the definition \eqref{hbz}
shows that as formal power series in $z$ we have 
\begin{align}
\hb^*(\z)\ \Phi^\pm_{p}\ \hb^*(\z)^{-1}
&=-\Phi^\pm_{p-1}\ z+(1-z^2)\sum_{k\ge p}\Phi^\pm_k\ z^{k-p}\,,
\label{hPhi}
\\
\hb^*(\z)\ \Psi^\pm_{p}\ \hb^*(\z)^{-1}
&=-\Psi^\pm_{p-1}\ z+(1-z^2)\sum_{k\ge p}\Psi^\pm_k\ z^{k-p}\,.
\label{hPsi}
\end{align}
The Fock vacuum expectation value of $\hb^*(\z)$ 
is easily evaluated, with the result
\begin{equation}
\langle \hb^*(\z)\rangle=1-z^2\,.
\label{expec}
\end{equation}
The operator $\hb^*(\z)$ is characterized by 
\eqref{hPhi}, \eqref{hPsi} and \eqref{expec}. 
We thus find the expression on $\cF$ 
\begin{align}
\hb^*(\z)&=(1-z^2)
\exp\left(\sum_{\nu=1}^\infty
\frac{z^\nu}{\nu}(\mathcal{I}_{-\nu}-\mathcal{I}_\nu)\right),
\label{formulah} 
\end{align}
where we used the standard bosonization formula
\begin{equation*}
\mathcal{I}_\nu=\sum_{\e=\pm}\sum_{p\in\Z}
:\Phi^\e_p\Psi^{-\e}_{p+\nu}:\,
\end{equation*}
which satisfies
\begin{align*}
&[\mathcal{I}_\mu,\mathcal{I}_\nu]=2\mu\delta_{\mu+\nu,0}\,,
\\
&\mathcal{I}_\nu(y^{2S(0)})=0\quad (\nu>0).
\end{align*}
Since the coefficients of the expansion of 
$\log\hb^*(\z)$ belong to the creation part of the 
Heisenberg algebra,  
the statement of Proposition is obvious.
\end{proof}

\begin{cor}\label{independence2}
The elements \eqref{monomial} are linearly independent. 
\end{cor}
\begin{proof}
This is a direct consequence of Lemma \ref{independence} and 
the commutation relations \eqref{Com1} and \eqref{Com3}.
\end{proof}


\subsection{Another family of operators}

In addition to $\bb_p,\cb_p$,  
we have considered in \cite{HGSII} 
another set of annihilation operators $\overline{\bb}_p,
\overline{\cb}_p$.  
They have the same support property
\be
&&\supp \overline{\xb}_p(X)\subset [k,m]
\quad (\overline{\xb}=\overline{\bb},\overline{\cb}),
\\
&&\overline{\xb}_p(X)=0\quad \text{if $p>m-k+1$} 
\quad (\overline{\xb}=\overline{\bb},\overline{\cb}), 
\en
anti-commute with $\bb_{p'},\cb_{p'}$ and satisfy 
\bea
&&[\bb^*_p,\overline{\bb}_{p'}]_+=
[\cb^*_p,\overline{\cb}_{p'}]_+=-\tb^*_{p-p'+1}\,.
\label{Com6}
\ena
In fact, they can be expressed
in terms of the operators \eqref{tbcbc}
as follows (see \cite{HGSIII}, Corollary A.2):
\bea
&&(\tb^*_1)^{-1}\overline{\bb}_p
=-\sum_{p'\ge0}\hb^*_{p'}\bb_{p+p'},
\quad
(\tb^*_1)^{-1}\overline{\cb}_p
=-\sum_{p'\ge 0}\hb^*_{p'}\cb_{p+p'}.
\label{bar1}
\ena
Due to \eqref{kill}, 
the sum in the right hand side is finite on each operand $X\in \cW^{(\al)}$.  
Even though these operators are not independent, 
we find it useful to take them into consideration. 

Let us introduce the creation counterpart to these operators.
Define $\overline{\bb}^*_p,\overline{\cb}^*_p$ ($p\ge 1$) inductively by 
\bea
\bb^*_p=\sum_{p'=1}^p\hb^*_{p-p'}\overline{\bb}^*_{p'},
\quad
\cb^*_p=\sum_{p'=1}^p\hb^*_{p-p'}\overline{\cb}^*_{p'}\,.
\label{bar2}
\ena
We have then 
\bea
&&[\overline{\bb}^*_p, (\tb^*_1)^{-1}\overline{\bb}_{p'}]_+=
[\overline{\cb}^*_p, (\tb^*_1)^{-1}
\overline{\cb}_{p'}]_+=-\delta_{p,p'}\,.
\label{Com5}
\ena
In the sequel we consider the generating series
\be
&&\xb(\z)=\sum_{p\ge 1}(\z^2-1)^{-p}\xb_p
\quad (\xb=\bb,\cb,\overline{\bb},\overline{\cb}), 
\\
&&\xb^*(\z)=\sum_{p\ge 1}(\z^2-1)^{p-1}\xb^*_p
\quad (\xb^*=\bb^*,\cb^*,\overline{\bb}^*,\overline{\cb}^*). 
\en
Then we have the relations of formal series
\be
&&\overline{\xb}(\z)\equiv 
-\tb_1^*\hb^*(\z)\xb(\z)\quad\bmod (\z^2-1)^0,
\\
&&\overline{\xb}^*(\z)=\hb^*(\z)^{-1}\xb^*(\z)\,.
\en

Unlike $\bb^*_p,\cb^*_p$, they enlarge the support essentially 
{\it to the left}. 

\begin{prop}
For any $X\in (\cW^{(\al)})_{[k,m]}$ we have 
\bea
&&\supp \overline{\xb}^*_p(X)\subset [k-p+1,m+1]
\quad (\overline{\xb}^*=\overline{\bb}^*,\overline{\cb}^*).
\label{supp3}
\ena
\end{prop}
\begin{proof}
It is known (see \cite{HGSII}, Lemma 3.7) 
that $\bb^*(\z)(X)$ has an expression 
\be
&&\bb^*(\z)(X)=
\lim_{l\to\infty}
\Rcb_{l,l-1}(\z^2)\cdots \cdots\Rcb_{m+2,m+1}(\z^2)
\Bigl(Y(\z)\Bigr)  
\en
for some $Y(\z)\in(\cW^{(\al)})_{[k,m+1]}$. 
Using the formula
\be
&&\hb^*(\z)^{-1}(X')=
\lim_{l\to\infty}
\Rcb_{-l+1,-l}(\z^2)^{-1}\cdots \cdots
\Rcb_{m'+1,m'}(\z^2)^{-1}(X')  
\en
valid for any $X'\in(\cW^{(\al)})_{[k',m']}$, 
we obtain that 
\be
&&\overline{\bb}^*(\z)(X)=
\lim_{l\to\infty}
\Rcb_{-l+1,-l}(\z^2)^{-1}\cdots \cdots
\Rcb_{k,k-1}(\z^2)^{-1}
\cdots \cdots\Rcb_{m+1,m}(\z^2)^{-1}
\Bigl(Y(\z)\Bigr)\,.
\en
By the same argument as for \eqref{hsupp},
this expression implies that the support property 
$\overline{\bb}^*_p(X)\subset [k-p+1,m+1]$ holds. 

The case of $\overline{\cb}^*_p$ is entirely similar.
\end{proof}

\medskip

\noindent{\it Remark.} 
To simplify the notation, in this paper we have modified the 
definition of the generating functions given in \cite{HGSII}. 
Denoting those in \cite{HGSII} by $\bb^*_{II}(\z)$, etc., 
the present definition is related to them as follows.
\begin{align*}
&\bb_{II}(\z)=\z^{-\al-\bS}\left(\bb_0+\bb(\z)\right),
\\
&\cb_{II}(\z)=\z^{\al+\bS}\left(\cb_0+\cb(\z)\right),
\\
&\bb^*_{II}(\z)=\z^{\al+\bS+2}\bb^*(\z),
\\
&\cb^*_{II}(\z)=\z^{-\al-\bS-2}\cb^*(\z).
\end{align*}
\qed

\section{A basis of $(\cW^{(\al)})_{[1,n]}$}\label{sec3}

The goal of this section is to  
construct a basis of $(\cW^{(\al)})_{[k,m]}$ for all intervals $[k,m]$ 
using linear combinations of \eqref{monomial}. 
In view of the translational invariance
we concentrate on the case $[1,n]$. 

In order to verify the spanning property, one has to 
find sufficiently many operators supported on $[1,n]$. 
In general the monomials \eqref{monomial}  have too large support, 
and suitable linear combinations of them must be chosen. 
We do this in two steps. First 
we introduce certain elements $B_J$ and show that their support 
is contained in $[1,n]$. 
We then construct general basis elements
by applying to them 
annihilation operators which do not enlarge the support. 

\subsection{Elements $B_J$}

Let $l$ be a non-negative integer satisfying $n\ge l\ge 0$. 
We define the numbers 
\bea
C^{i_1,\cdots,i_l}_{j_1,\cdots,j_l;k_1,\cdots,k_l}
\quad (1\le i_p,j_p,k_p)
\label{CIJK}
\ena 
by the generating series
\bea
&&\frac{\Delta(x)\Delta(y)\Delta(z)}
{\prod_{i,j=1}^l(1-x_iy_j)(1-x_iz_j)}
=
\sum C^{i_1,\cdots,i_l}_{j_1,\cdots,j_l;k_1,\cdots,k_l} 
\prod_{p=1}^l (x_p^{i_p-1}y_p^{j_p-1}z_p^{k_p-1}).
\label{Cauchy}
\ena
Here 
$x=(x_1,\cdots,x_l)$, $y=(y_1,\cdots,y_l)$, $z=(z_1,\cdots,z_l)$, 
$\Delta(x)=\prod_{1\le i<j\le l}(x_i-x_j)$, 
and the sum is taken over all positive integers 
$i_p,j_p,k_p$ ($p=1,\cdots,l$).  

When all sequences are decreasing ($i_1>\cdots>i_l$, etc.), 
we identify them with subsets 
$I=\{i_1,\cdots,i_l\} \subset [1,n]$, etc., 
and write $C^I_{J,K}$ for \eqref{CIJK}. 
We shall write $|I|$ for the cardinality of $I$. 
In this case 
$C^I_{J,K}$ coincides with the Littlewood-Richardson coefficient
$c^{\lambda}_{\mu,\nu}$ well-known in combinatorics 
(see \cite{Mac}, eq.(5.2)); in particular they are non-negative integers. 
The precise correspondence reads
\be
C^{I}_{J,K}=c^{\la(I)}_{\la(J),\la(K)}, 
\en
where $\la(I)=(\la_1,\cdots,\la_l)=(i_1-l,\cdots,i_l-1)$.

In what follows, for a subset $I=\{i_1,\cdots,i_l\}$ ($i_1>\cdots>i_l$)
of $[1,n]$, we write 
\be
&&\xb_M=\xb_{i_1}\cdots\xb_{i_l}
\en
for $\xb=\bb,\cdots,\overline{\cb}^*$.

Now we introduce a family of operators $B_J$ indexed by $J\subset [1,n]$.
Writing $l=|J|$ we define
\bea
&&B_J=\sum_{I,K}C^I_{J,K}
\bb^*_{n}\cdots \overset{i_1}{\cb^*_{k_1}}\cdots 
\overset{i_l}{\cb^*_{k_l}}\cdots \bb^*_1(q^{2\al S(0)})\,.
\label{BJ}
\ena
In the last line, the sum is taken over all subsets
$I,K\subset [1,n]$ with $I=\{i_1,\cdots,i_l\}$, 
$K=\{k_1,\cdots,k_l\}$, and 
$\cb^*_{k_p}$ is placed at the $i_p$-th slot. 

Since $c^{\la}_{\mu,\emptyset}=\delta_{\la,\mu}$, we have
\be
B_J=\pm \bb^*_{\complement J}\cb^*_{\{1,\cdots,l\}}(q^{2\al S(0)})+\cdots
\en
where $\cdots$ stands for terms which do not contain $\cb^*_{\{1,\cdots,l\}}$. 
\medskip

\noindent {\it Example.}
Let $n=4$ and $l=2$. Suppressing $q^{2\al S(0)}$ we have 
\be
&&B_{\{3,4\}}=\cb^*_2\cb^*_1\bb^*_2\bb^*_1\,,
\\
&&B_{\{2,4\}}=\cb^*_2\bb^*_3\cb^*_1\bb^*_1
+\cb^*_3\cb^*_1\bb^*_2\bb^*_1
\,,
\\
&&B_{\{1,4\}}=\cb^*_2\bb^*_3\bb^*_2\cb^*_1
+\cb^*_3\bb^*_3\cb^*_1\bb^*_1
+\cb^*_4\cb^*_1\bb^*_2\bb^*_1\,,
\\
&&B_{\{2,3\}}=\bb^*_4\cb^*_2\cb^*_1\bb^*_1
+\cb^*_3\bb^*_3\cb^*_1\bb^*_1
+\cb^*_3\cb^*_2\bb^*_2\bb^*_1
\,,
\\
&&B_{\{1,3\}}=\bb^*_4\cb^*_2\bb^*_2\cb^*_1+
\cb^*_3\bb^*_3\bb^*_2\cb^*_1+
\bb^*_4\cb^*_3\cb^*_1\bb^*_1+
\cb^*_3\bb^*_3\cb^*_2\bb^*_1+
\cb^*_4\bb^*_3\cb^*_1\bb^*_1+
\cb^*_4\cb^*_2\bb^*_2\bb^*_1
\,,
\\
&&B_{\{1,2\}}=
\bb^*_4\bb^*_3\cb^*_2\cb^*_1
+\bb^*_4\cb^*_3\bb^*_2\cb^*_1
+\cb^*_4\bb^*_3\bb^*_2\cb^*_1
+\bb^*_4\cb^*_3\cb^*_2\bb^*_1
+\cb^*_4\bb^*_3\cb^*_2\bb^*_1
+\cb^*_4\cb^*_3\bb^*_2\bb^*_1
\,.
\en
\medskip

These elements admit an alternative expression in terms 
of the other set of creation operators. 
 
\begin{lem} We have 
\bea
B_J=\sum_{I,K}C^I_{J,K}
\overline{\bb}^*_{n}\cdots \overset{i_1}{\overline{\cb}^*_{k_1}}\cdots 
\overset{i_l}{\overline{\cb}^*_{k_l}}\cdots \overline{\bb}^*_1(q^{2\al S(0)})\,.
\label{BJ2}
\ena
\end{lem}
\begin{proof}
First we note that the definition \eqref{BJ} can be written as 
\be
&&B_J
=\frac{1}{l!^2}
\Bigl(\sum_{n\ge i_p,k_p\ge 1\atop 1\le p\le l}
C^I_{J,K}\cb^*_{k_1}\bb_{i_1}\cdots \cb^*_{k_l}\bb_{i_l}\Bigr)
B_\emptyset,
\\
&&B_\emptyset=\bb^*_n\cdots\bb^*_1(q^{2\al S(0)}). 
\\
\en
Actually the restriction $n\ge i_p,k_p$ is irrelevant, since otherwise the 
corresponding term is zero. 
Extending the suffix of $B_J$ by anti-symmetry, we consider their 
generating series. Inserting
\be
\cb^*_k\bb_i=\oint\frac{d\z^2}{2\pi i}\frac{d\xi^2}{2\pi i}
\cb^*(\z)\bb(\xi)(\z^2-1)^{-k}(\xi^2-1)^{i-1}
\en
and using \eqref{Cauchy}, we get 
\be
&&l!^2\sum_{j_1,\cdots,j_l\ge 1}
B_{j_1,\cdots,j_l}\prod_{p=1}^ly_p^{j_p-1}
=\oint\prod_{p=1}^l\frac{d\z^2_p}{2\pi i}\frac{d\xi^2_p}{2\pi i}
\prod_{p=1}^lz_p 
\frac{\Delta(x)\Delta(y)\Delta(z)}
{\prod_{i,j=1}^l(1-x_iy_j)(1-x_iz_j)}
\\
&&\qquad\qquad\qquad\quad\times 
\cb^*(\z_1)\bb(\xi_1)\cdots\cb^*(\z_l)\bb(\xi_l)B_\emptyset\,.
\en
Here we have set $z_p=(\z^2_p-1)^{-1}$, $x_p=\xi_p^2-1$, and 
the integral is taken along the contour encircling 
$\z^2_p=\xi^2_p=1$, $|\z^2_p-1|>|\xi^2_{p'}-1|$. 
Noting that 
\be
\frac{\Delta(x)\Delta(z)}{\prod_{i,j=1}^l(1-x_iz_j)}=
\det\left(\frac{1}{\z^2_p-\xi_{p'}^2}\right)_{1\le p,p'\le l}
\en
and integrating over $\xi^2_p$, we find 
\bea
&&(-1)^l\frac{l!}{\Delta(y)}
\sum_{j_1,\cdots,j_l\ge 1}B_{j_1,\cdots,j_l}
\prod_{p=1}^ly_p^{j_p-1}
\label{BJint}\\
&&=\oint\prod_{p=1}^l\frac{d\z^2_p}{2\pi i}
\prod_{i,j=1}^l\frac{1}{1-y_i(\z^2_j-1)}\ 
\cb^*(\z_1)\bb(\z_1)\cdots \cb^*(\z_l)\bb(\z_l)B_\emptyset\,.
\nn
\ena
On the other hand, we have
\be
&&\cb^*(\z)\bb(\z)\equiv -\overline{\cb}^*(\z)
(\tb^*_1)^{-1}\overline{\bb}(\z)\quad \bmod (\z^2-1)^0,
\\
&&\bb^*_n\cdots\bb^*_1=\overline{\bb}^*_n\cdots\overline{\bb}^*_1\,.
\en
Noting further that 
$[\overline{\bb}^*_p,(\tb^*_1)^{-1}\overline{\bb}_{p'}]_+=-\delta_{p,p'}$, 
we see from \eqref{BJint} that the definition \eqref{BJ} of $B_J$
is unchanged if we interchange barred and unbarred operators. 
\end{proof}

\begin{cor}\label{cor:BJ}
We have
\be
\supp B_J\subset [1,n]. 
\en
\end{cor}
\begin{proof}
This follows from the two different expressions 
\eqref{BJ},\eqref{BJ2} and the 
support property \eqref{supp1},\eqref{supp3}:
\be
\supp B_J\subset [1,\infty)\cap (-\infty,n]=[1,n].
\en
\end{proof}

\subsection{Construction of a basis}

Consider the following set of elements of $\cW^{(\al)}$:
\bea
\overline{\bb}_M\cb_N(B_J)
\quad
(J\subset [1,n],; M\subset [1,n-|J|],\ N\subset [1,|J|]).
\label{base}
\ena
As $M,N,J$ vary there are altogether $4^{n}$ such elements. 
Due to Corollary \ref{cor:BJ}, and 
since annihilation operators preserve the support, 
they are all supported in $[1,n]$. 
\medskip

\noindent{\it Example.}
We omit writing $q^{2\al S(0)}$.
For $n=2$, we have
\be
B_\emptyset=\bb^*_2\bb^*_1,\ 
B_{\{2\}}=\cb^*_1\bb^*_1,\ 
B_{\{1\}}=\bb^*_2\cb^*_1+\cb^*_2\bb^*_1,\ 
B_{\{1,2\}}=\cb^*_2\cb^*_1. 
\en
The elements \eqref{base} give
\be
&& \bb^*_2\bb^*_1;\\
&& \bb^*_2,\ \bb^*_1,\ \tb^*_1\bb^*_2-\tb^*_2\bb^*_1,\ \tb^*_1\bb^*_1;\\
&& 1,\ \tb^*_1,\ \tb^*_2,\ 
(\tb^*_1)^2,
\ \cb^*_1\bb^*_1,\ \bb^*_2\cb^*_1+\cb^*_2\bb^*_1;
\\
&& \cb^*_2,\ \cb^*_1,\ \tb^*_1\cb^*_2-\tb^*_2\cb^*_1,\ \tb^*_1\cb^*_1;
\\
&& \cb^*_2\cb^*_1\,.
\en
\medskip

The following is the main result of this note. 

\begin{thm}
The elements \eqref{base} 
constitute a basis of 
$(\cW^{(\al)})_{[1,n]}$.
\end{thm}
\begin{proof}
It is sufficient to show that the set \eqref{base} is linearly independent. 
Suppose there is a linear relation
\be
\sum_{M,N,J}A_{M,N,J}\overline{\bb}_M\cb_N(B_J)=0
\en
with some scalars $A_{M,N,J}$. 
The sum is taken over $J\subset [1,n]$, 
$M\subset[1,n-|J|]$, $N\subset[1,|J|]$. 
The left hand side is a linear combination of monomials 
containing $l=|J|$ number of $\bb^*$ and 
$\tb^*$. 
Hence the sum is separately zero for each fixed $l$. 

Fixing $M_0\subset[1,n-l]$, $N_0\subset[1,l]$, we apply  
$\overline{\bb}_{[1,\cdots,n-l]\backslash M_0}\cb_{[1,l]\backslash N_0}$ to 
both sides. 
We obtain
\be
&&0=\sum_{|J|=l}A_{M_0,N_0,J}\overline{\bb}_{\{1,\cdots,n-l\}}\cb_{\{1,\cdots,l\}}
\\
&&\quad\times
\sum_{I,K}C^{I}_{J,K}
\cb^*_{k_l}\bb_{i_l}\cdots
\cb^*_{k_1}\bb_{i_1}\bb^*_n\cdots\bb^*_1(q^{2\al S(0)}).
\en
In the second sum only the term with 
$K=\{1,\cdots,l\}$ contributes.
Noting that $C^{I}_{J,\{1,\cdots,l\}}=\delta_{I,J}$, we obtain 
\be
0=\sum_{J}A_{M_0,N_0,J}
\det\left(\hb^*_{j'_a-b}\right)_{1\le a,b\le n-l}
(q^{2\al S(0)}),
\en
where $\{j'_1,\cdots,j'_{n-l}\}$ denotes the complement of 
$J=\{j_1,\cdots,j_n\}$ in $[1,n]$. 
By Lemma \ref{independence}, 
Schur functions in $\hb^*_p$'s applied to 
$q^{2\al S(0)}$ are linearly independent. 
Hence we conclude that $A_{M_0,N_0,J}=0$ for all $J$. 
This completes the proof. 
\end{proof}
\bigskip

{\it Acknowledgements.}\quad
HB is grateful to the Volkswagen Foundation and to the
'Graduiertenkolleg' DFG project:
"Representation theory  and its application in
mathematics and physics" for the financial support.
Research of MJ is supported by the Grant-in-Aid for Scientific 
Research B-20340027 and B-20340011. 
Research of TM is supported by
the Grant-in-Aid for Scientific Research B--17340038.
Research of FS is supported by  EC networks   "ENIGMA",
contract number MRTN-CT-2004-5652
This work was supported by World Premier International
Research Center Initiative (WPI Initiative), MEXT, Japan.

\end{document}